   \documentclass[12pt,single-space,onecolumn]{IEEEtran}
\usepackage{amsmath,amssymb,mathrsfs,graphicx}
\usepackage{epsfig,subfigure}
\allowdisplaybreaks[2]
\newcommand{\bbr}{\mathbb R}

\newtheorem{theorem}{Theorem}[section]
\newtheorem{lemma}{Lemma}[section]

\newtheorem{proposition}{Proposition}[section]

\newtheorem{definition}{Definition}[section]

\ifCLASSINFOpdf
\else
\fi
\begin{document}
\title{Cucker-Smale flocking with alternating leaders}

\author{Zhuchun Li  and Seung-Yeal Ha
         % <-this % stops a space
\thanks{
Z. Li is with the Department of Mathematics, Harbin Institute of Technology, Harbin 150001,
 China  (e-mail: lizhuchun@gmail.com). S.-Y. Ha is with the Department of Mathematical Sciences,
Seoul National University, Seoul 151-747, Korea  (e-mail: syha@snu.ac.kr).
This work was carried out when Z. Li visited Seoul National University; he appreciated the support from PARC-SNU.
}% <-this % stops a space
%\thanks{Manuscript received April 19, 2005; revised January 11, 2007.}
}

\maketitle

\begin{abstract}
%\boldmath
We study the emergent flocking behavior in a group of Cucker-Smale flocking agents under rooted leadership with alternating leaders. It is well known that the network topology regulates the emergent  behaviors of flocks. All existing results on the Cucker-Smale model with leader-follower topologies assume a fixed leader during temporal evolution process. The rooted leadership is the most general topology taking a  leadership. Motivated by collective behaviors observed in the flocks of birds, swarming fishes and  potential engineering applications, we consider the rooted leadership with alternating leaders; that is, at each time slice there is a leader but it can be switched among the agents from time to time. We will provide several sufficient conditions leading to the asymptotic flocking among the Cucker-Smale agents under rooted leadership with alternating leaders.
\end{abstract}

\begin{IEEEkeywords}
Cucker-Smale agent, alternating leaders, rooted leadership, flocking
\end{IEEEkeywords}

% \fi
\IEEEpeerreviewmaketitle

\section{Introduction}
The purpose of this paper is to study the emergent flocking phenomenon to the generalized Cucker-Smale (C-S)
 model with alternating leaders. Roughly speaking,
 the terminology ``flocking'' represents the phenomena that autonomous agents,
 using only limited environmental information,
%and metric and topological interactions,
organize into an ordered motion, e.g., flocking of birds, herds of cattle,  etc.  These  collective  motions have gained increasing interest from  the research communities
in biology, ecology, sociology and  engineering \cite{A-B-H-K-L,  C-S-D, J-L-M, L-P-L-S, P-E-G,  T-T, V-C-B-C-S} %B-B, B-B-B-N, T-B,
 due to their various applications in  sensor
networks, formation of robots and spacecrafts, financial markets and opinion formation
in social networks.
%Among many flocking models, our interest lies on the Cucker-Smale model \cite{C-S2} which
%describes the evolution of a flock.

In \cite{C-S2, C-S1}, Cucker and Smale proposed a nonlinear second-order model to  study the emergent behavior of flocks. Let $x_i, v_i \in \bbr^d $ be the position and velocity of the $i$-th C-S agent, and $\psi_{ij}\geq 0$ be interaction weight between $j$ and $i$-th agents. Then, the discrete-time C-S model reads as
 \begin{align}
\begin{aligned} \label{c-s-1}
x_i(t+1)&=x_i(t)+hv_i(t), \quad  i=1,2,\dots,N,\\
 v_i(t+1)&=v_i(t)+h\sum\limits_{j=1}^N \psi_{ij}(x(t))\left[v_j(t)-v_i(t)\right],\\
  \psi_{ij}(x(t))&=\frac{1}{(1+|x_i(t)-x_j(t)|^2)^\beta}, \quad\,\, \beta \geq 0,
\end{aligned}
\end{align}
where $h$ is a time-step. For \eqref{c-s-1}, the ``{\it asymptotic flocking}'' means
\[  \sup_{t \in \mathbb N} |x_i(t)-x_j(t)|< \infty, \quad \lim_{t\to\infty} (v_i(t)-v_j(t))= 0,  \quad \mbox{$\forall~i\neq j$}. \]
The study of flocking behavior of multi-agent system based on mathematical models dates back to \cite{J-L-M,V-C-B-C-S}, even before Cucker and Smale. However, the significance of C-S model lies on the solvability of the model and phase-transition like behavior from the unconditional flocking to conditional flocking, as the decay exponent $\beta$ increases from zero to some number larger than $\frac{1}{2}$.

 %Other multi-particle systems with
%   different choices of   interaction weights can be found in \cite{B-C, J-L-M,V-C-B-C-S}.
    %For example, as can be seen in a recent experiment \cite{B-C},
%    the communication weight $\psi_{ij}$ can be dependent on the topological distance.
 %rather than metric distance, say
 %i.e., some fixed finite number of closest neighboring agents.
  The particular choice of weight function $\psi_{ij}$ in \eqref{c-s-1} is a crucial ingredient which
  makes this model attractive.
  %; the main feature is, the convergence results depend on conditions on the
%initial state only.
We note that in the original C-S model \eqref{c-s-1}, the agents are interacting under the all-to-all distance depending couplings $\psi_{ij}=\psi_{ji}>0$ for all $i\neq j.$
%For the C-S model, the term  ``flocking'' means that   moving particles (such as birds, fish, etc) exhibit the following time-asymptotic  behavior:
%\begin{enumerate}[(1)]
%\item Form a group and move in formation:  $x_i(t)-x_j(t)$   has a limit for all $i,j=1,2,\dots,N$;
%    \item Have velocity alignment: $v_i(t)-v_j(t)\to 0$ as $t\to \infty.$
%\end{enumerate}
%In the analytic study of C-S flocking, the critical threshold
%phenomena appears depending on the size of $\beta$, i.e., the decay
%rate in the communication weight $\psi_{ij}$. More specifically,  Cucker and Smale \cite{C-S2} showed that for $\beta< 1/2$,
%the system \eqref{c-s-1} exhibits flocking behavior unconditionally, i.e.,
%regardless of initial configurations, %asymptotic flocking is guaranteed,
%whereas for $\beta \geq 1/2$, the emergence of flocking behavior is conditional
%on the size of  the initial configurations.
Later, Cucker-Smale's results were extended in several directions, e.g., stochastic noise effects \cite{A-H,C-M,H-L-L},
 collision avoidance \cite{A-C-H-L,C-D2}, steering toward preferred
 directions \cite{C-H},  bonding forces \cite{P-K-H}, %H-H-K,
 and mean-field limit \cite{C-F-R-T,H-Liu,H-T}.   % F-H-T, D-F-T, B-C-H-K,
 In particular,
 %the C-S flocking mechanism was applied to a control law in order to make
% various spacecrafts form a formation \cite{P-E-G}.
  an unexpected application was proposed by Perea, G\'{o}mez  and   Elosegui \cite{P-E-G} who suggested to use the C-S flocking mechanism \cite{C-S2}  in the formation of spacecrafts  for the Darwin space mission. Recently, the C-S
flocking mechanism was also applied to the modeling of emergent cultural classes in sociology and the stochastic volatility in financial markets \cite{A-B-H-K-L, H-K-L, K-H-J-K}.

In this paper, we consider  the  Cucker-Smale flocking  under a switching of   leadership topology with alternating leaders.
%In 2007, Cucker and Smale \cite{C-S2} proposed a mathematical model to describe the
%emergence of flocking in the natural animal behavior, using an all-to-all interaction
%between the particles.
It is well known that the interaction topology is an important component to understand the dynamics of multi-particle systems and vice versa. Biological complex systems are ubiquitous in our nature and indeed take various interaction topologies. The first work in relation with the C-S model other than all-to-all topology is due to J. Shen,  who introduced the C-S model under hierarchical leadership \cite{She}.  A more general topology with leadership including hierarchical one was introduced  by Li and Xue in \cite{L-X}, namely, the rooted leadership. Unfortunately, the analysis given in \cite{L-X} cannot be applied to the continuous-time C-S model in a general setting. The continuous-time C-S model with a rooted leadership was studied in the framework of fast-slow dynamical systems in \cite{H-L-S-X} for some restricted situation. Recently, a topology with   joint rooted leadership was also considered in \cite{L-H-X}, in which  a ``joint'' connectivity is imposed only along some time interval, instead of every time slices.  Note that in previous works \cite{L-X, L-H-X} involving interaction topologies with leadership, the leader agent is assumed to be fixed in temporal evolution of flocks. This is not realistic. We can often observe that the leaders in migrating flocks can be changed during their migration. For example, as the large flocks of birds make a long journey from continent to continent, the leader birds located in the front of the flock endure larger resistance from the neighboring airs, e.g., wind. So leaders have to spend more energy than other followers. To save the energy of leader birds, leaders change alternatively. Of course, we can also find alternating leaders in our human social systems, for example, the periodic election of political leaders.  Motivated by these situations, we study the asymptotic flocking behavior of the C-S model with alternating leaders.

For the flocking analysis of the C-S model, most existing studies assume all-to-all and symmetric couplings so that the conservation of momentum is guaranteed, which is crucial in the energy estimates \cite{C-S2,C-S1,H-Liu}. In contrast, when the interaction topology is not symmetric, there is no general systematic approach for flocking estimate. The induction method is applied to hierarchical leadership \cite{C-D1,She}, which relies on the triangularity of the adjacency matrix. Another useful tool, the self-bounding argument developed by Cucker-Smale in \cite{C-S2,C-S1}, can be applied to different topologies; however,  it requires a flocking estimate  that relies on the topologies. For all-to-all coupling, the estimate is made on the matrix 2-norm through the spectrum of symmetric graph Laplacian. For rooted leadership, the authors in \cite{L-X,L-H-X} employed the $(sp)$ matrices \cite{X-G,X-L} to study the infinity norm of a reduced Laplacian. Note that for all these special cases, the asymptotic velocity for flocking is a priori known, either the mean value of the initial state or just that of the leader. Thus, we can study the dynamics of newly defined variables, i.e., the fluctuations around the average velocity, or the states relative to the fixed leader, which can be bounded,  to study the flocking behavior.  However, in the case of alternating leaders, we do not have the accurate information on the asymptotic velocity of the flock and the dynamics of referenced variables (see \eqref{reference}) cannot be given by nonnegative matrices as in \cite{L-X,L-H-X}.  To overcome this difficulty, we consider the combined dynamics of the original system and  the reference  system. We employ  the estimates in \cite{C-M-A1, C-M-A2} for the first-order consensus problem to find a priori estimate for the original system. From this, we can estimate the evolution of referenced velocity to support the self-bounding argument.

 The rest of this paper is organized as follows. In Section 2,
we   describe  our model and present a consensus estimate that  is useful in this work.
%and we also   review   the C-S flocking model.
In Section 3, we  provide the flocking estimates for the discrete-time C-S model under
  rooted leadership with alternating leaders. 
 In Section 4, we present some numerical simulations.  
  Finally, Section 5 is devoted to the summary of this paper.

{\bf Notation:} Given $x\in \mathbb R^N$, we use the notations $|x|_\infty$ and  $|x|$ to denote the infinity norm (maximum norm) and 2-norm (Euclidean norm) of the vector   respectively.  For a $N \times N$ matrix $A \in \mathbb R^{N\times N}$, we use   $\|A\|_\infty$
to denote the infinity norm, that is, the maximum absolute row sum of $A$, and for two $N\times N$ matrices  $A=(a_{ij})$ and $B=(b_{ij})$, we use $A\circ B$ to denote the element-wise product, i.e., $A\circ B=(a_{ij}b_{ij}).$
%We also use $(\,\cdot\,)_{ij}$ to denote the element of a matrix in the $i$th row and $j$th column.

\section{Preliminaries}
In this section we introduce the C-S flocking model under rooted leadership with
 alternating leaders. A useful estimate for the ``flocking'' matrix
  will be presented as well.
 \subsection{Flocks with alternating leaders}
In this subsection, we present a brief description of the C-S flocking model
under rooted leadership with alternating leaders.

Consider a group of agents $\{ 1,2,\dots,N\}$.  For the description of interaction topology, we use  graph theory \cite{D} as follows. A digraph ${\mathcal{G}} =({\mathcal{V}},{\mathcal{E}} )$ (without self-loops)   representing $N$ particles with interactions,  is defined by
\[ {\mathcal{V}} :=\{1, 2,\dots,N\}, \qquad  {\mathcal{E}} \subseteq {\mathcal{V}} \times {\mathcal{V}}\setminus  \{(i,i):~i\in {\mathcal V}\}. \]
We say $(j,i)\in {\mathcal{E}} $ if and only if $j$ is a neighbor of $i$, i.e., $j$ influences $i$.  As an information flow chart, we may write
 $j \to i$ if and only if $(j,i)\in {\mathcal{E}}$.
A directed path from $j$ to $i$ (of length $n+1$) comprises a sequence
of distinct arcs of the form $j\to k_1\to k_2\to\dots\to k_n\to i$.
%The distance from $j$ to $i$ is the length of the shortest path from $j$ to $i$.

On the other hand, once the directed neighbor graph
 ${\mathcal{G}} =( {\mathcal{V}}, {\mathcal{E}})$ is chosen,
  the associated adjacency matrix, denoted by $\chi( {\mathcal G})=(\chi_{ij})$, is
   given by
\begin{equation*}\label{chi}
\chi_{ij}=\left\{\begin{array}{c}
1, \qquad  \mbox{if}~~ (j,i)\in  {\mathcal{E}},\\
0, \qquad  \mbox{if}~~ (j,i)\notin  {\mathcal{E}}.\\
\end{array}\right.
\end{equation*}
%we use $\chi=(\chi_{ij})$ to denote the interaction pattern. Precisely, we put $\chi_{ij}=1$ if the agent $i$  is influenced by $j$,  otherwise we put $\chi_{ij}=0$.
Then, the C-S model on the digraph graph $ {\mathcal G}$ is given by
 \eqref{c-s-1} with the second equation replaced by
%\begin{equation}\label{C-S1}\left\{\begin{array}{c}
%x_i(t+1)=x_i(t)+hv_i(t),\quad\,\qquad  i=0,1,\dots,N,\\
% v_i(t+1)=v_i(t)+h\sum\limits_{j=0}^N \chi_{ij}\psi_{ij}(x(t))\left(v_j(t)-v_i(t)\right).\\
%\end{array}\right.\end{equation}
\begin{align*}\begin{aligned}%\label{C-S0}
%x_i(t+1)&=x_i(t)+hv_i(t),\qquad  i=1,2,\dots,N, ~t\in \mathbb N,\\
%v_0(t+1)&=v_0(t),\\
 v_i(t+1)&=v_i(t)+h\sum\limits_{j=1}^N \chi_{ij} \psi_{ij}(x(t))
 \left[v_j(t)-v_i(t)\right].%\\
 %\quad  i=1,2,\dots,N,
%   \psi_{ij}(x(t))&=\frac{1}{(1+|x_i(t)-x_j(t)|^2)^\beta}, \quad \beta \geq 0.
\end{aligned}\end{align*}
Thus, there is an interaction from $j$ to $i$ with strength  $\psi_{ij}(x(t))$ as long as it exists.
%We next provide a graph theoretic interpretation of   system \eqref{C-S1} by using    digraphs without self-loops.
%A directed graph, or digraph for short,
%is a pair $\mathcal{G}=(\mathcal{V}, \mathcal{E})$ of:
%\begin{enumerate}[(1)] \item a set $\mathcal{V}=\{0, 1,\cdots,N\}$, whose elements are
%called vertices or nodes;
%\item a set $\mathcal{E}\subseteq \mathcal{V}\times \mathcal{V}$ of
%ordered pairs of distinct vertices, called arcs or directed edges.
% \end{enumerate}
% review a graph theoretic interpretation of the Vicsek model \eqref{e1}
In order to take the interaction weight $\psi_{ij}(x)$ into account,
we refer to the matrix $\chi \circ \Psi_x:=\big(\chi_{ij}\psi_{ij}(x(t))\big)$
%, where  $\Psi_x:=\psi_{ij}(x)$,
as the weighted adjacency matrix
of the C-S system   on the digraph  ${\mathcal{G}}$.

Below, we use the symbol $\mathcal{I}$ to denote a finite set indexing all admissible digraphs
${\mathcal{G}}_p=({\mathcal{V}}_p,{\mathcal{E}}_p)$ and let $\sigma:\mathbb N\to \mathcal I$ be a switching signal.  At each time point $t\in \mathbb N$, the system is registered on an admissible  graph ${\mathcal G}_{\sigma(t)}$, and  thus has a weighted adjacency matrix given by $\chi^{\sigma(t)}\circ \Psi_x$.    In this setting, we write the system  as the C-S system undergoing
switching  of the neighbor graphs with a switching signal  $\sigma $:
%\begin{equation}\label{switchsignal}
%\chi := \chi^{\sigma (t)}=(\chi_{ij}^{\sigma(t)})\,,\qquad
%\sigma:\mathbb N\to \mathcal{I} ~ \mbox{is ~a ~ switching~signal.}
%\end{equation}
\begin{align}\begin{aligned}\label{C-S1}
x_i(t+1)&=x_i(t)+hv_i(t),\qquad  i=0,1,\dots,N, \\
%v_0(t+1)&=v_0(t),\\
 v_i(t+1)&=v_i(t)+h\sum\limits_{j=1}^N \chi_{ij}^{\sigma(t)} \psi_{ij}(x(t))\left[v_j(t)-v_i(t)\right],\\
 %\quad  i=1,2,\dots,N,
   \psi_{ij}(x(t))&=\frac{1}{(1+|x_i(t)-x_j(t)|^2)^\beta}.
\end{aligned}\end{align}
% Thus the system \eqref{C-S1}-\eqref{switchsignal}
%is on a dynamic graph, or with switching interaction patterns.
%For
%$p\in {\mathcal I}$, we use
%$\mathcal{G}_p=({\mathcal{V}_p},{\mathcal{E}_p})$ to denote the proper
%subgraph of ${\mathcal{G}}_p$ by deleting vertex 0 and all edges issued from the vertex 0.
%Next we introduce three classes of matrices associated to the digraphs ${\mathcal G}_p$
%and $\bar{\mathcal G}_p$.

We now introduce the  definition  of   C-S model under rooted leadership  with alternating leaders. For this we first present the rooted leadership  \cite{L-X}.
\begin{definition}\label{definealterleader}
%\emph{(Joint rooted leadership)}
\label{definejrl}
\begin{enumerate}
\item
The system \eqref{C-S1} is under \emph{rooted leadership  at time $t$}, if for the  digraph
$  {\mathcal{G}}_{\sigma(t)} $, there exists a unique vertex, say $r_t\in \mathcal V$, such that  the vertex $r_t$ does not have an incoming path from others, but any other vertex in
$\mathcal V$ has a directed path from $r_t$. The vertex $r_t$ represents the leader in the flock.
\item
The system \eqref{C-S1} is under \emph{rooted leadership with alternating leaders}, if
 the system \eqref{C-S1} is under rooted leadership  at each time slice, but the leader $r_t$ is not fixed for all time.
\end{enumerate}
\end{definition}
Note that the leader can be changed from time to time; thus, the asymptotic velocity is not a priori known, even if the flocking can be achieved.

%In Fig. \ref{figjointleadership} we show how several neighbor-graphs form a joint rooted leadership.
%The main results on the flocking can be summarized as follows: Let $T_0$ be the bound of lengths of the time-intervals
%for the joint rooted leadership topology; then, for $\beta< 1/(2NT_0)$,  flocking
%emerges unconditionally. In contrast, for $ \beta\geq 1/(2NT_0)$, flocking emerges
%when the initial configuration satisfies some conditions (see Theorem \ref{thmconsensus}).
%For C-S model under strongly joint rooted leadership (see Definition \ref{definejrl}), the
%unconditional flocking can be established for $\beta< 1/(2LT_0)$, where $L\,(\leq N)$ is a constant   depending on
%the ``joint'' connectivity topology (see Theorem \ref{thmconsensus2}).

%The proof of our main result relies upon
% several new estimates on  $(sp)$ matrices. In order to present the proof, in the next section we will focus on the $(sp)$ matrices.

\subsection{Consensus estimates}\label{subsec} In this subsection, we present a convergence  estimate  given in \cite{C-M-A1} for the first-order consensus problem.
 Given a sequence of stochastic matrices (also known as Markov matrices \cite{L-R}) $F_1, F_2, \dots \in \mathbb R^{N\times N}$,
 for consensus we expect that the product $F_t\cdots F_2 F_1$ converges
to a rank one matrix, i.e., has the same row vectors. For a single stochastic matrix $F$, under some connectivity condition of its associated graph, the matrix iteration $F^k$ converges to the rank one matrix $\mathbf 1 \pi$ with $\pi$ being the left-eigenvector of $F$, i.e., $\pi F=\pi$. To deal with the case of time-dependent state transition matrices, we  introduce some notations following \cite{C-M-A1}. Let $F$ be a stochastic matrix and we denote by $\lfloor F\rfloor$ the row vector  whose $j$th element is the smallest  element of the $j$th column of $S$. Let
$ [F] =F- \mathbf 1 \lfloor F\rfloor,$ then we have $\|[F]\|_\infty=1-\lfloor F\rfloor \mathbf 1,$ where $\mathbf 1=(1,1,\dots,1)^{\mathrm T}.$
In some sense, $[F]$ measures how much the matrix $F$ is different with a rank one matrix.
It is obvious that  a product of  stochastic matrices must be a stochastic matrix.
For an  infinite sequence of stochastic matrices $F_1, F_2, \dots,$ the limit
$$\lfloor\cdots F_t\cdots F_2F_1\rfloor:=\lim_{t\to \infty}\lfloor F_t\cdots F_2F_1\rfloor $$ always exists \cite{C-M-A1}, even if the product $F_t\cdots F_2F_1$
itself does not have a limit.  In order to form a consensus,
we expect the product $F_t\dots F_2F_1$  to converge to a rank one stochastic matrix, i.e.,
a  matrix of the form  $\mathbf 1 c$. If this is true, then the limit must be $\mathbf 1  \lfloor\cdots F_t\cdots F_2F_1\rfloor$.
In the following, we will say that the matrix product $F_t\dots F_2F_1$ converges to $\mathbf 1 \lfloor\cdots F_t\cdots F_2F_1\rfloor$ exponentially fast at a rate no slower than $\lambda$ if there exist nonnegative constants $b$ and $\lambda<1$ such that \begin{equation}\big\|F_t\cdots F_2F_1-\mathbf1 \lfloor\cdots F_t\cdots F_2F_1\rfloor\big\|_\infty \leq b \lambda^t, \quad t\geq 1.\end{equation}
We write $A\geq B$ if $A-B$ is a nonnegative matrix. %, and $A> B$ if $A-B$ is a positive matrix.
 The following result    gives a sufficient condition to the exponential convergence.

%$\llbracket \rrbracket$ $\textlbrackdbl \textrbrackdbl$
\begin{proposition}\label{propconv}\cite{C-M-A1}
 (1) For any pair of stochastic matrices $F_1$ and $F_2$, we have $$[F_2F_1]\leq [F_2][F_1];$$

\noindent (2) Let $b$ and $\lambda<1$ be nonnegative constants. Suppose that $F_1, F_2, \dots$ is an infinity sequence of stochastic matrices with $$ \| [F_t\cdots F_2F_1]\|_\infty \leq b\lambda^t,\,\, t\geq 0.$$
Then, the matrix product $F_t\dots F_2F_1$ converges to $\mathbf 1 \lfloor\cdots F_t\cdots F_2F_1\rfloor$ exponentially fast at a rate no slower than $\lambda$.
\end{proposition}

Therefore, if each of the matrices in the sequence $F_1, F_2, \dots$ satisfies $\|[F_i]\|_\infty\leq \lambda,$ then $F_t\dots F_2F_1$ does  converge  to $\mathbf 1 \lfloor\cdots F_t\cdots F_2F_1\rfloor$ exponentially. Since $\|[F]\|_\infty=1-\lfloor F\rfloor \mathbf 1,$ $\|[F]\|_\infty<1$ if and only if the matrix $\lfloor F\rfloor $ has at least one positive element, i.e., the matrix $F$ has at least one nonzero column.
For a stochastic matrix $F=(F_{ij})\in \mathbb R^{N\times N}$, we define the associated digraph as $\mathcal G=(\mathcal V, \mathcal E)$ with $\mathcal V=\{1,2,\dots,N\}$ and $\mathcal E=\{(j,i): F_{ij}>0\}.$
We call a   graph $\mathcal G$ is a strongly rooted graph if there exists some vertex $j$ such that $(j,i)\in \mathcal E$ for all $i\neq j.$ For such a $j$, we say that it is  a strong root of the graph and the graph is strongly rooted at $j$. In what follows, we also say $\mathcal G$ is a rooted graph if there exists a vertex which has a path to any other agent; such a vertex is called the root of the graph. Then we arrive at the following result.

\begin{lemma}\label{lemmaeq}\cite{C-M-A1}
The  digraph  associated to a stochastic matrix $F$ is strongly rooted if and only if $\|[F]\|_\infty<1$.
\end{lemma}
We next introduce the composition of digraphs. By the composition of   digraphs $\mathcal G_p$ with $ \mathcal G_q$,  denoted by $\mathcal G_q \circ \mathcal G_p$, we mean the digraph with the
vertex set $\mathcal V$  and arc set defined in such a way so that $(i, j)$ is an arc of the
composition just in case there is a vertex $k$ such that $(i, k)$ is an arc of $\mathcal G_p$ and $(k, j)$
is an arc of $\mathcal G_q$. Denote their associated flocking matrices by $F_p$ and $F_q$ respectively.  Then, we see that  the flocking matrix associated to the digraph $\mathcal G_q \circ \mathcal G_p$ is exactly the matrix product  $F_qF_p$.
\begin{proposition}\label{propcomp} \cite{C-M-A1} Let $\mathcal G_{\sigma(1)}, \mathcal G_{\sigma(2)}, \dots,$ be a sequence of rooted graph, then  for any $t_1\in \mathbb N$, the graph $\mathcal G_{\sigma(t_1+ (N-1)^2)
}  \circ \cdots \circ\mathcal G_{\sigma(t_1+ 2)}\circ \mathcal G_{\sigma(t_1+ 1)}$ is a strongly rooted graph.
\end{proposition}
Based on this result, we can obtain a strongly rooted graph from the composition of rooted leadership with alternating leaders.

\section{Flocking analysis}
In this section, we introduce the C-S flocking matrix and a reference system to study    the C-S model with alternating leaders.

\subsection{A reference system and the flocking matrix}
We consider a group of particles $\{ 1,2,\dots,N\}$ whose dynamics is governed by \eqref{C-S1}. Let $x=(  x_1,
x_2,\dots, x_N)^\top$ and $v=( v_1,
v_2,\dots, v_N)^\top\in {(\mathbb{R}^3)}^{N }$ be the position
and velocity vector of the flock, respectively. In order to  simplify the notations, for a given solution $\{(x(t), v(t))\}$ to system \eqref{C-S1}  under a switching signal $\sigma$,
  we write
\begin{equation}\label{d}
\psi_{ij}(t):=\psi_{ij}(x(t)),\quad\,\,
 d_i(t):=\sum\limits_{j=1, j\neq i}^N \chi_{ij}^{\sigma(t)}\psi_{ij}(t).
 \end{equation}
In order to use a self-bounding argument, we introduce a reference system for the $N$-flocks. We use the last agent  as the reference and set \begin{align}\begin{aligned}\label{reference}  \hat x &:=(\hat x_1,  \dots, \hat x_{N-1})^\top=
(x_1-x_N,   \dots, x_{N-1}-x_N)^\top, \\
\hat v &:=(\hat v_1,  \dots, \hat v_{N-1})^\top=(v_1-v_N,   \dots,
v_{N-1}-v_N)^\top.
\end{aligned}\end{align}
It is obvious that the asymptotic flocking behavior is equivalent to  the boundedness of $\hat x$ together with the zero convergence of $\hat v$.
If we set $\hat  x_{N}=0$ and $\hat v_{N}=0$, then
$$|x_i-x_j|^2=|\hat x_i-\hat x_j|^2
\leq 2(|\hat x_i|^2+|\hat x_j|^2)\leq 2|\hat x|^2, \quad 1\leq i,j\leq N.$$ This means that the C-S communication weights satisfy
\begin{align}\label{psi}\psi_{ij}(t)=\frac{1}{(1+|x_i(t)-x_j(t)|^2)^\beta}
\geq \frac{1}{(1+2|\hat x(t)|^2)^\beta}.\end{align}
Note that   the dynamics of $\hat x(t)$ and ${\hat v}(t)$ are governed by
\begin{align}\label{eqhatx}
\hat x(t+1)=\hat x(t) +h\hat v(t), \quad {\hat v}(t+1)= P_{\sigma(t)}\hat v(t),
\end{align}
where the matrix $P_{\sigma(t)}\in \mathbb R^{(N-1)\times (N-1)}$  is given by \begin{equation*}
(P_{\sigma(t)})_{ij}=\left\{\begin{array}{c}
1-hd_i-h\chi^{\sigma(t)}_{Ni}\psi_{Ni}(t), \qquad\qquad   i=j,\\
  h\chi^{\sigma(t)}_{ij}\psi_{ij}(t)-h\chi^{\sigma(t)}_{Nj}\psi_{Nj}(t), \qquad  i\neq j,\\
\end{array}\right.
\end{equation*}
where $i,j=1,2,\dots,N-1.$
  If the flock has a fixed leader, say $N$, then $\chi_{Nj}^{\sigma(t)}\equiv0$ for all $j\neq N$,  and thus the matrix $P_{\sigma(t)}$  is a nonnegative matrix  provided a sufficiently small $h$.  However, if the leader agent changes from time to time, we cannot expect  the matrix $P_{\sigma(t)}$  to be always a nonnegative matrix. Therefore, the approach in \cite{L-X,L-H-X} cannot be applied in this case.
  To carry out a flocking estimate, we will not use the explicit difference equation of $\hat v(t)$,  instead  we   derive a direct estimate of $\hat v(t)$ through   $v(t)$.    This is the key idea in this study  apart from the previous  works.

%  the position  $x(t)$
%$$|x_i-x_j|^2
%\leq 2(|x_i|^2+|x_j|^2)\leq 2|x|^2,$$ so we have
%\begin{displaymath}\psi_{ij}(t)=\frac{1}{(1+|x_i(t)-x_j(t)|^2)^\beta}
%\geq \frac{1}{(1+2|x(t)|^2)^\beta}.\end{displaymath}

To estimate $v(t)$, we  derive a compact form  from the system  \eqref{C-S1} as follows:
%\begin{equation}\left\{\begin{array}{c} \label{reducedflock}
%x(t+1)=x(t)+hv(t), \quad\\
% v(t+1)=(Id-hL_{\sigma(t)})v(t). \\
%\end{array}\right.\end{equation}
\begin{align}\begin{aligned}  \label{reducedflock}
x(t+1)&=x(t)+hv(t),  \\
 v(t+1)&=(Id-hL_{\sigma(t)})v(t)=:F_{\sigma(t)}v(t), \end{aligned}\end{align}
 where $L_{\sigma(t)}\in \mathbb R^{N\times N}$ is the weighted  Laplacian of digraph $\mathcal G_{\sigma (t)}$, that is,
\begin{align*}\begin{aligned}&L_{\sigma(t)}= \left(\begin{array}{cccc}
 d_1(t) & -\chi_{12}^{\sigma(t)}\psi_{12}(t) & \cdots & -\chi_{1N}^{\sigma(t)}\psi_{1N}(t)\\
 -\chi_{21}^{\sigma(t)}\psi_{21}(t) & d_2(t) & \cdots & -\chi_{2N}^{\sigma(t)}\psi_{2N}(t)\\
 \cdots & \cdots & \ddots & \cdots\\
 -\chi_{N1}^{\sigma(t)}\psi_{N1}(t) & -\chi_{N2}^{\sigma(t)}\psi_{N2}(t) & \cdots & d_N(t)\\
\end{array}
\right).\end{aligned}\end{align*}
The matrix $F_{\sigma(t)}:=Id-hL_{\sigma(t)}$ is called the (C-S) flocking matrix at time $t$ associated to the neighbor graph $\mathcal G_{\sigma (t)}$.
If we choose a small $h>0$ such that all diagonal entries of $F_{\sigma(t)}$ are nonnegative,
then $F_{\sigma(t)}$ is a stochastic matrix, i.e., a nonnegative matrix with
 each row sum being 1.
%\section{Flocking behavior of Cucker-Smale system under joint rooted leadership}
%\setcounter{equation}{0}
%\subsection{An estimate}
%We recall that under the
%Cucker-Smale's settings, the connection weight is given by
%\begin{equation}\label{cs}\psi_{ij}(x)=\frac{1}{(1+|x_i-x_j|^2)^{\beta}}.
%\end{equation}

\subsection{Basic estimates} In this subsection we present an estimate of $\hat v(t)$ through $v(t)$. We first concentrate on the flocking matrix $F_{\sigma(t)}$ that determines the dynamics of $v(t)$. To do this, we will employ the estimates in \cite{C-M-A1} (see Subsection \ref{subsec}).

Note that for flocking under rooted leadership (see Definition \ref{definealterleader}), the neighbor graph is a rooted graph with the leader agent acting as the root. Thus,
Proposition \ref{propcomp} and Lemma \ref{lemmaeq} imply that the product of $(N-1)^2$ flocking matrices must  satisfy  $$\big\|[F_{\sigma(t_1+ (N-1)^2-1)}\cdots F_{\sigma(t_1+ 1) }F_{\sigma(t_1 )}]\big\|_\infty<1.$$
Inspired by this fact, we will present an estimate for the product of flocking matrices which can be applied to the analysis of the second-order C-S model.

We first estimate the convergence of $F_{\sigma(t)}\cdots F_{\sigma(1) }F_{\sigma(0)}$ by Proposition \ref{propconv}.

\begin{proposition}\label{propestimate}
Suppose that $\{(x(t),v(t))\}_{t\in \mathbb N}$ is a solution of C-S model \eqref{C-S1} with alternating leaders.
Assume \begin{equation}\label{boundx}|\hat x(t)|\leq B,\quad t\in \mathbb N,\end{equation}
and
\begin{equation}\label{h}
h<\frac{1}{N+1}.\end{equation}
  Then, we have
\begin{align}\begin{aligned}  \label{eqprop}
 &\big\|F_{\sigma(t)}\cdots F_{\sigma( 1) }F_{\sigma(0)}-\mathbf 1 \lfloor \cdots F_{\sigma(t)}\cdots F_{\sigma( 1) }F_{\sigma(0)} \rfloor\big\|_\infty  \leq
\left(1-(hR)^{(N-1)^2}\right)^{\big[\frac{t+1}{(N-1)^2}\big]},
\end{aligned}\end{align}
where $R=\frac{1}{(1+2B^2)^{\beta}}$.
\end{proposition}
\begin{proof}
%Let $\{(x(t),v(t))\}_{t\in \mathbb N}$ be a solution of C-S model \eqref{C-S1} with alternating leader.
%Assume that we have a uniform
%bound on the referenced position:
According to \eqref{psi} and \eqref{boundx}, for any $t\geq 0$, we have
\begin{equation}\label{boundpsi}\frac{1}{(1+2B^2)^\beta}
\leq \psi_{ij}(t)\leq 1,\qquad 1\leq i,j\leq N.\end{equation}
Moreover, by the assumption \eqref{h}, we have
\begin{equation}\label{Pdiag1}(F_{\sigma(t)})_{ii}=1-hd_i(t)\geq \frac{1}{N+1}>\frac{h}{(1+2B^2)^\beta}, \quad 1\leq i\leq N.\end{equation}
From \eqref{boundpsi} and \eqref{Pdiag1} we see that under the assumptions \eqref{boundx} and \eqref{h},   $F_{\sigma(t)}$ is a stochastic matrix with  nonzero entries  not less than $hR:=\frac{h}{(1+2B^2)^\beta}$. Consequently,  all the nonzero elements of the matrix product $F_{\sigma(t)}\cdots F_{\sigma(1)} F_{\sigma(0)}$ must be not less than $(hR)^{t+1}$.
Note that if the system \eqref{C-S1} is under rooted leadership at time $t$, then the neighbor graph $\mathcal G_{\sigma(t)}$ is a rooted graph with the leader acting as the root. We recall Proposition \ref{propcomp} to see that the composition of neighbor graph along a time interval of length $(N-1)^2$ must be a strongly rooted graph. By Lemma \ref{lemmaeq}, this means that  any $(N-1)^2$-product of flocking matrices satisfies
  \[\big\|[F_{\sigma(t_1+ (N-1)^2-1)}\cdots F_{\sigma(t_1+1) }F_{\sigma(t_1)}]\big\|_\infty<1,\]
or equivalently,  \[1-\lfloor F_{\sigma(t_1+ (N-1)^2-1)}\cdots F_{\sigma(t_1+1) }F_{\sigma(t_1)}\rfloor \mathbf 1<1.\] Thus, the matrix   $F_{\sigma(t_1+ (N-1)^2-1)}\cdots F_{\sigma(t_1+1) }F_{\sigma(t_1)}$ has at least one nonzero column.       Because all of  the  nonzero elements of $F_{\sigma(t_1+ (N-1)^2-1)}\cdots F_{\sigma(t_1+1) }F_{\sigma(t_1)}$ are not less than $(hR)^{(N-1)^2},$   we find that \[1-\lfloor F_{\sigma(t_1+ (N-1)^2-1)}\cdots F_{\sigma(t_1+1) }F_{\sigma(t_1)}\rfloor \mathbf 1\leq 1-(hR)^{(N-1)^2},\]
i.e., \begin{equation}
\big\|[F_{\sigma(t_1+ (N-1)^2-1)}\cdots F_{\sigma(t_1+1) }F_{\sigma(t_1)}]\big\|_\infty\leq1-(hR)^{(N-1)^2}.
\end{equation}
This implies that for all $t\in \mathbb N$, \begin{equation}  \label{eq}
\|F_{\sigma(t)}\cdots F_{\sigma( 1) }F_{\sigma(0)}\|_\infty \leq
\left(1-(hR)^{(N-1)^2}\right)^{\big[\frac{t+1}{(N-1)^2}\big]}.
\end{equation}
We now combine \eqref{eq} and Proposition \ref{propconv} to obtain \eqref{eqprop}.
\end{proof}
Next, we use Proposition \ref{propestimate} to estimate the evolution of  $\hat v(t).$
 \begin{proposition}\label{propestimate1}
Suppose   $\{(x(t),v(t))\}_{t\in \mathbb N}$ is a solution of C-S model \eqref{C-S1} with alternating leaders. If (\ref{boundx}) and (\ref{h}) hold,  then for the referenced velocity $\hat v(t)$  we have
\begin{equation}  \label{eqprop1}
|\hat v(t)|_\infty\leq  2\left(1-(hR)^{(N-1)^2}\right)^{\big[\frac{t+1}{(N-1)^2}\big]} |v(0)|_\infty.
\end{equation}
%where $R=\frac{1}{(1+2B^2)^{\beta}}$.
\end{proposition}
\begin{proof} For the simplicity of notation,  we denote the asymptotic  velocity alignment state for the solution $\{(x(t),v(t))\}$  as   $$v^\infty=\mathbf 1 \lfloor \cdots F_{\sigma(t)}\cdots F_{\sigma(1) }F_{\sigma(0)} \rfloor v(0).$$
Then, it follows from \eqref{eqprop} that we have  \begin{align*} | v(t)- v^\infty |_\infty &=\big|F_{\sigma(t)}\cdots F_{\sigma( 1) }F_{\sigma(0)}v(0)-\mathbf 1 \lfloor \cdots F_{\sigma(t)}\cdots F_{\sigma(1) }F_{\sigma(0)} \rfloor v(0) \big|_\infty \\ &\leq   \left(1-(hR)^{(N-1)^2}\right)^{\big[\frac{t+1}{(N-1)^2}\big]} |v(0)|_\infty.\end{align*}
Due to the definition of referenced variables $(\hat x(t), \hat v(t))$, we easily find
\begin{align*}|\hat v(t)|_\infty&\leq \big|\big(v(t)-v^\infty\big)-\big(v_N(t)-v^\infty\big)\big|_\infty\\&\leq 2\left(1-(hR)^{(N-1)^2}\right)^{\big[\frac{t+1}{(N-1)^2}\big]} |v(0)|_\infty.\end{align*}
 Here, we use $v_N(t)$ itself to denote the $N$-duplication of $v_N(t)\in \mathbb R^3$.
\end{proof}

\subsection{Flocking behavior}
To carry out the flocking analysis, we use the spectral norm (or 2-norm) of vectors and matrices. However,   the previous  estimates for $v(t)$ is given by the infinity norm. Due to the equivalence of norms in finite-dimensional space,  there
exists a constant $\lambda\geq1,$ such that for all $\hat v\in
(\mathbb{R}^3)^{N-1},$
\begin{equation}\label{lambda}
|\hat v|_\infty\leq |\hat v|\leq \lambda|\hat v|_\infty,
\end{equation}
where $|\cdot|$ denotes the 2-norm of vector.
Before we present the main result, we quote an elementary
lemma from \cite{C-S2} without the proof. Consider the following algebraic equation:
\begin{equation}\label{Alg}
F(z) :=z^r-c_1z^s-c_2 = 0.
\end{equation}
\begin{lemma}\label{selfbdd}\cite{C-S2}
Suppose that the coefficients and exponents in $F$ satisfy
\[ c_1, c_2>0 \quad \mbox{and} \quad r>s>0. \]
Then the equation \eqref{Alg} has a unique positive zero $z_*$ satisfying
\[ z_* \leq \max\{(2c_1)^{\frac{1}{r-s}},
(2c_2)^{\frac{1}{r}}\}, \quad \mbox{and} \quad F(z)\leq 0 \,\,\,\, \mbox{for\,\, $0\leq z \leq z_*$}. \]
\end{lemma}
We are now ready to present our main result   on the Cucker-Smale flocking with alternating leaders.

\begin{theorem}
\label{thmconsensus}  %Assume that for some constants
%$H,\sigma>0$ and $\beta\geq0$,
 %\begin{displaymath}\psi_{ij}=\frac{H}{(\sigma^2+|x_i-x_j|^2)^\beta}.\end{displaymath}
Let the discrete-time Cucker-Smale model \eqref{C-S1} be under   rooted leadership  with alternating leaders.
  Assume that the time step $h$
fulfills (\ref{h}), and one of the following three hypothesis holds:
\begin{enumerate}%[(1)]
\item $s<1;$ %\quad |x(0)| < \infty, \qquad |v(0)| < \infty.$
\item $s=1$ %$\beta=\frac{1}{2L}$
 and \[ |v(0)|<\frac{h^{(N-1)^2-1}}{2\sqrt2(N-1)^2\lambda};\] % |x(0)| < \infty, \qquad
\item $s>1$, and
\begin{eqnarray} \label{assumption1}&&\left(\frac{1}{ {a}}\right)^{\frac{1}{s-1}}
\left[\left(\frac{1}{s}\right)^{\frac{1}{s-1}}-
\left(\frac{1}{s}\right)^{\frac{s}{s-1}}\right]- {b}>
\frac{8 \lambda^2| v(0)|^2s^{\frac{1}{s-1}}}{N^2}
 {a}^{\frac{1}{s-1}}
+\frac{4\sqrt{2}\lambda| v(0)|}{N}. \end{eqnarray}
 Here, \begin{equation}\label{abs}a=2\sqrt{2}\lambda h^{-(N-1)^2+1}
(N-1)^2| v(0)|, \quad
b=1+\sqrt{2}|\hat x(0)|, \quad   s= 2\beta (N-1)^2.
\end{equation} % and
%$$\Lambda=\left(\frac{1}{s \mathbf{a}}\right)^{\frac{1}{s-1}}.$$
\end{enumerate} Then the system \eqref{C-S1} or \eqref{reducedflock} has a time-asymptotic flocking:
\begin{enumerate}  \item there exists a constant $B_0>0$ such that $|\hat x(t)|\leq B_0, ~ \forall\,t\in \mathbb{N};$  \item $\hat v(t)$ exponentially converges
to zero as $t\rightarrow \infty$. \end{enumerate} Moreover, there exists $\hat{x}^\infty\in
{(\mathbb{R}^3)}^{N-1}$ such that $\hat x(t)\longrightarrow \hat{x}^\infty$ as
$t\rightarrow \infty.$
\end{theorem}

%In order to prove Theorem \ref{consensuswithfiexedtopology}, we need
%to estimate the convergence rate of the matrix iterated product. As
%we showed above, the $P_n$'s are  $(sp)$ matrices with fixed $(sp)$
%subscript sets. Looking into the proof of Lemma 2.1 in \cite{xue07},
%we know that any product of $L$ $P_n$'s, say
%$$P_{L-1}P_{L-2}\cdots P_0,$$
%is a trivial $(sp)$ matrix, that is, all its row sums are strictly
%less than 1. Moreover, we have the following estimation on the
%infinity norm of their iterated product.
% We first give a proposition and a lemma that are needed for the proof of
% Theorem \ref{consensuswithfiexedtopology}.

% Towards the proof of the above theorem we benefit a lot from the
% technique developed by Cucker and Smale in
%\cite{CuckerSmale07a}. In order to make this paper self-contained,
%we prefer to show the entire proof.
% The following lemma is from \cite{CuckerSmale07a}.
%
%\begin{lemma}\label{lemma} Let $c_1, c_2>0$ and $r>s>0$. Then the equation
%$$F(z)=z^r-c_1z^s-c_2=0$$  has a unique positive zero $\hat{z}$. In
%addition $$\hat{z}\leq \max\{(2c_1)^{\frac{1}{r-s}},
%(2c_2)^{\frac{1}{r}}\}$$ and $F(z)\leq 0$ for $0\leq z \leq
%\hat{z}.$
%\end{lemma}

%Next we give a proof for Theorem
%\ref{thmconsensus}. %by applying the self-bounding
%technique developed by Cucker and Smale in \cite{CuckerSmale07a}. In
%order to make this paper self-contained, the entire proof is
%presented next.

\begin{proof}
Fix a discrete-time mark $T\in \mathbb{N}$ and define
\begin{equation}\label{nstar}
|\hat x_*|=\max\limits_{0\leq t\leq T} |\hat x(t)|, \,\,\,\,\,\,
 T_*\in \mathrm{argmax}_{0\leq t\leq T}|\hat x(t)|.
\end{equation}
  Then by Proposition \ref{propestimate1}, the estimate \eqref{eqprop1} holds for newly defined
 $R:=(1+2{|\hat x_*|}^2)^{-\beta}$, as long as we restrict  $t$
 within $[0,T]$. That is,
 \[
|\hat v(t)|_\infty\leq  2\left(1-(hR)^{(N-1)^2}\right)^{\big[\frac{t+1}{(N-1)^2}\big]} |v(0)|_\infty, \quad t\in [0,T].
\]
   By \eqref{lambda} we  find
\begin{align*}|\hat v(t)|\leq \lambda|\hat v(t)|_\infty&\leq   2\lambda\left(1-(hR)^{(N-1)^2}\right)^{\big[\frac{t+1}{(N-1)^2}\big]} |v(0)|_\infty\\&\leq 2\lambda\left(1-(hR)^{(N-1)^2}\right)^{\big[\frac{t+1}{(N-1)^2}\big]} |v(0)| ,\,\,\,\,\,\,t\in[0,T].\end{align*}
By the dynamics of referenced position $\hat x(t)$, i.e., \eqref{eqhatx}, for $t\in[0,T],$
we have
\begin{displaymath}\begin{array}{ll}
|\hat x(t)|&\leq |\hat x(0)| +h\sum\limits_{\tau=0}^{t-1}|\hat v(\tau)|\\
&\leq |\hat x(0)| +h\sum\limits_{\tau=0}^{t-1} 2\lambda\big(1-(hR)^{(N-1)^2}\big)^{\big[\frac{\tau+1}{(N-1)^2}\big]} |  v(0)| \\
&\leq |\hat x(0)|+2h \lambda |  v(0)| \sum\limits_{\tau=0}^{\infty}
\big(1-(hR)^{(N-1)^2}\big)^{\big[\frac{\tau+1}{(N-1)^2}\big]}\\
 &\leq |\hat x(0)|+2h \lambda
| v(0)|(N-1)^2 \sum\limits_{\tau=0}^{\infty} \left(1-(hR)^{(N-1)^2}\right)^\tau\\
&\leq |\hat x(0)|+2h \lambda |  v(0)|(N-1)^2 (hR)^{-(N-1)^2}\\
&= |\hat x(0)|+ \frac{\sqrt 2}{2}a(1+2|\hat x_*|^2)^{\beta (N-1)^2}.
\end{array}\end{displaymath}
 \\In particular, for $t=T_*,$ we have
%\begin{equation}\label{e1}
\begin{displaymath}|\hat x_*|\leq |\hat x(0)| +\frac{\sqrt{2}}{2}
a(1+2|\hat x_*|^2)^{\beta (N-1)^2}.\end{displaymath}
%\end{equation}
Let $Z:=(1+2|\hat x_*|^2)^{\frac{1}{2}},$ then the above relation and \eqref{abs} give
\begin{equation}\label{selfbd}
Z\leq 1+\sqrt{2}|\hat x_*|\leq a Z^{2\beta (N-1)^2}+ b.
\end{equation}
In order to apply Lemma \ref{selfbdd} we define $F(z)$ as follows:  $$F(z)=z-az^{s}-b, \qquad s=2\beta (N-1)^2.$$
(1) Assume $s<1$. The relation \eqref{selfbd} says that
$F(Z)\leq 0$. By appealing to Lemma \ref{selfbdd} we have $Z\leq U_0$ with
$U_0\leq \max\{(2 {a})^{\frac{1}{1-s}}, 2 {b}\}$, and
then $|\hat x_*|\leq
\big(\frac{{U_0}^2-1}{2}\big)^{\frac{1}{2}}$. Note that
$U_0$  depends only on $ {a}$ and $ {b}$, which are
independent of the pre-assigned time $T$. Therefore, the bound
$|\hat x(T)|\leq |\hat x(T_*)|\leq B_0:=
\big(\frac{{U_0}^2-1}{2}\big)^{\frac{1}{2}}$ is uniform
for all $T=0, 1, 2,\dots$. That is, $|\hat x(t)|\leq B_0$ for all $t\in
\mathbb{N}$.

 Now, we choose
$B:=B_0$, %:=\left(\frac{{U_0}^2-\sigma^2}{2}\right)^{\frac{1}{2}},
and  $R:=(1+2B_0^2)^{-\beta}$ and
 use Proposition \ref{propestimate1}  for any time $t\in \mathbb N$ to find
  that $\hat v(t)$ exponentially converges to $0$.   Finally, for
all $t_2> t_1$, we have
\begin{displaymath}\begin{array}{ll}
|\hat x(t_2)-\hat x(t_1)|&\leq \sum\limits_{\tau=t_1}^{t_2-1}|\hat x(\tau+1)-\hat x(\tau)|\leq
h \sum\limits_{\tau=t_1}^{t_2-1}|\hat v(\tau)|\\&\leq
2h\lambda| v(0)| \sum\limits_{\tau=t_1}^{t_2-1} \Big(1-(hR)^{(N-1)^2}\Big)^{\big[\frac{\tau+1}{(N-1)^2}\big]}\\&\leq 2h\lambda| v(0)|
\sum\limits_{\tau=t_1}^{\infty}\Big(1-(hR)^{(N-1)^2}\Big)^{\big[\frac{\tau+1}{(N-1)^2}\big]}\\
&\leq 2h^{1-{(N-1)^2}}R^{-{(N-1)^2}}\lambda {(N-1)^2}|  v(0)|
\Big(1-(hR)^{{(N-1)^2}} \Big)^{[\frac{t_1}{{(N-1)^2}}]}.
\end{array}\end{displaymath}
Note that the right-hand side tends to zero as $t_1\to +\infty$ and is independent of $t_2$, by Cauchy principle we
deduce that there exists some $\hat{x}^\infty\in {(\mathbb R^3)}^{N-1}$ such that
$\hat x(t)\rightarrow \hat{x}^\infty$ as $t\to +\infty$.

\noindent (2) Assume $s=1$. Then (\ref{selfbd}) becomes
$$Z\leq  {a} Z+ {b},$$
which implies that $$|\hat x_*|\leq
\frac{\sqrt2}{2}\Big(\big(\frac{b}{1-a}\big)^2-1\Big)^{1/2}.$$ The
right-hand side  is positive by our hypothesis and thus gives a uniform
bound for $\hat x(t)$.  For the remaining parts we proceed as in case (1).
\begin{figure}[h]
%\nopagenumber
%\renewcommand{\baselinestretch}{1.0} inner_synchronized(u=0.98,c=0,0.2,1,5)
%\hfill
\begin{center}
\includegraphics[width=3.5in]{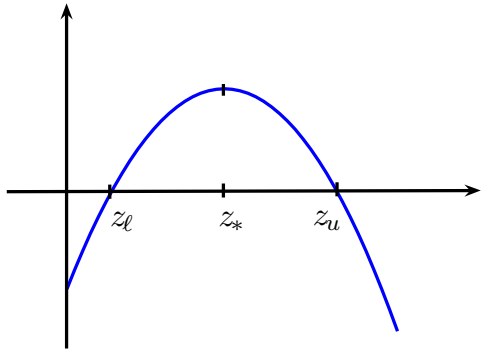}
\end{center}
\caption{Shape of $F$.} \label{shape}
\end{figure}

(3) Assume that $s>1$. The derivative $F'(z)=1-s {a}z^{s-1}$
has a unique zero at $z_*=(\frac{1}{s {a}})^{\frac{1}{s-1}}$
and
\begin{displaymath}\begin{array}{ll}F(z_*)=&\left(\frac{1}{s {a}}\right)^{\frac{1}{s-1}}- {a}
\left(\frac{1}{s {a}}\right)^{\frac{s}{s-1}}- {b}> 0,
\\\end{array}\end{displaymath}
by our hypothesis (\ref{assumption1}). Since $F(0)=- {b}<0$
and $F(z)\rightarrow -\infty$ as $z\rightarrow +\infty$ we see that
the shape of $F$ is as  in Fig.\ref{shape}. For $t\in \mathbb{N}$,
let $Z(t)=(1+2|\hat x(t_*)|^2)^{\frac{1}{2}},$ where  $t_*$ is
defined as in (\ref{nstar}), i.e., $t_*\in \mathrm{argmax}_{0\leq \tau\leq t}|\hat x(\tau)|$. For $t=0$ we must have $t_*=0$ and
\begin{displaymath}\begin{array}{ll}Z(0)=
(1+2|\hat x(0)|^2)^{\frac{1}{2}}\leq
1+\sqrt{2}|\hat x(0)|= {b}\leq
(\frac{1}{s {a}})^{\frac{1}{s-1}}=z_*.\end{array}\end{displaymath}
This implies that $Z(0)\leq z_\ell.$ Assume that there exists $t\in
\mathbb{N}$ such that $Z(t)\geq z_u$ and let $t_0$ be the first
such $t$. Then ${t_0}_*=t_0 \geq 1$ and for all $t<t_0,$
\begin{displaymath}\begin{array}{ll}
(1+2|\hat x(t)|^2)^{\frac{1}{2}}\leq z_\ell \leq
z_*,\end{array}\end{displaymath} that is,
\begin{displaymath}\begin{array}{ll}
|\hat x(t)|\leq
\big(\frac{{z_*}^2-1}{2}\big)^{\frac{1}{2}}.\end{array}\end{displaymath}
In particular,
\begin{displaymath}\begin{array}{ll} |\hat x(t_0-1)|^2\leq
\frac{{z_\ell}^2-1}{2}\leq
\frac{{z_*}^2-1}{2}.\end{array}\end{displaymath}
 On the other hand, $Z(t_0)\geq z_u$ gives
\begin{displaymath}\begin{array}{ll} |\hat x(t_0)|^2\geq
\frac{{z_u}^2-1}{2}\geq
\frac{{z_*}^2-1}{2}.\end{array}\end{displaymath} Thus we have
\begin{equation}\label{ap1} |\hat x(t_0)|^2-|\hat x(t_0-1)|^2\geq
\frac{{z_*}^2-{z_\ell}^2}{2}\geq
\frac{1}{2}(z_*-z_\ell)z_*.\end{equation}
 By the
intermediate value theorem, there is a $\xi\in [z_\ell, z_*]$ such
that $F(z_*)=F'(\xi)(z_*-z_\ell).$ Note that $F'(\xi)=1-s {a}\xi^{s-1} \in [0,1]$,
therefore,
\begin{displaymath}\begin{array}{ll}
z_*-z_\ell\geq F(z_*).\end{array}\end{displaymath}
We combine  (\ref{ap1}) to obtain
\begin{equation}\label{ap2}
|\hat x(t_0)|^2-|\hat x(t_0-1)|^2\geq \frac{1}{2}z_*F(z_*).\end{equation}
However, we have
\begin{displaymath} |\hat x(t_0)|-|\hat x(t_0-1)|\leq
|\hat x(t_0)-\hat x(t_0-1)|=h |\hat v(t_0-1)|\leq 2h \lambda|  v(0)|_\infty\leq 2h \lambda|  v(0)|.\end{displaymath}
Therefore,
\begin{displaymath}\begin{array}{ll}
|\hat x(t_0)|^2-|\hat x(t_0-1)|^2&=\left(|\hat x(t_0)|-|\hat x(t_0-1)|\right)^2+
2\left(|\hat x(t_0)|-|\hat x(t_0-1)|\right)|\hat x(t_0-1)|\\&\leq
4h^2\lambda^2|  v(0)|^2+4h\lambda| v(0)||\hat x(t_0-1)|\\&\leq
4h^2\lambda^2| v(0)|^2+4h\lambda| v(0)|\left(\frac{{z_*}^2-1}{2}\right)^{\frac{1}{2}}.
\end{array}\end{displaymath}
We combine this inequality with (\ref{ap2}) to obtain
\begin{displaymath}\begin{array}{ll}
z_*F(z_*) \leq
8h^2\lambda^2| v(0)|^2+8h\lambda| v(0)|\left(\frac{{z_*}^2-1}{2}\right)^{\frac{1}{2}}\leq
8h^2\lambda^2| v(0)|^2+4\sqrt2 h\lambda| v(0)|z_*,
\end{array}\end{displaymath} thus, we have
\begin{eqnarray} F(z_*)&&=\left(\frac{1}{ {a}}\right)^{\frac{1}{s-1}}\left[\left(\frac{1}{s}\right)^{\frac{1}{s-1}}-
\left(\frac{1}{s}\right)^{\frac{s}{s-1}}\right]- {b}\nonumber\\&&\leq
8h^2\lambda^2|  v(0)|^2
\left(s {a}\right)^{\frac{1}{s-1}}+4\sqrt{2}h\lambda| v(0)|\nonumber\\&&\leq
\frac{8 \lambda^2| v(0)|^2s^{\frac{1}{s-1}}}{N^2}
 {a}^{\frac{1}{s-1}}
+\frac{4\sqrt{2}\lambda| v(0)|}{N},\nonumber\end{eqnarray}
where (\ref{h}) is
used. This contradicts to our hypothesis (\ref{assumption1}).  So we
conclude that, for all $T\in \mathbb{N}$, $Z(t)\leq z_\ell$ and
then $|\hat x(t)|\leq
\big(\frac{{z_*}^2-1}{2}\big)^{\frac{1}{2}},$ which is a
uniform bound for $\hat x(t)$. Again we can proceed as in case (1) to complete the proof.
\end{proof}

\section{Simulations}
In this section, we present some numerical simulations to show the flocking behavior in a C-S model with alternating leaders.  For these simulations, we choose a flock consisting of three agents, labeled by $\mathcal V=\{1,2,3\}$, and we take the parameters \[h=0.2, \quad\mathrm{and}\quad \beta=\frac{1}{4}.\]  We consider the switching in three interaction topologies  described by   graphs
$$\mathcal G_1=\big(\mathcal V, \{(1,2), (1,3)\}\big), \,\, \mathcal G_2=\big(\mathcal V,  \{(2,1), (2,3) \}\big),   \,\,\mathcal G_3=\big(\mathcal V,  \{(3,1), (3,2) \}\big).$$
This means that in $\mathcal G_i$, the agent $i$ acts as the leader and there are information flows from $i$ to the other two agents. We choose the switching signal $\sigma(t)$ as \begin{equation*}\sigma(t)=(t,~\mod 3)+1, \qquad t \,\,\textrm{is   the  discrete time mark}.\end{equation*}
In other words, the  sequence of neighbor graphs is given by \begin{equation}\label{switchsig}\mathcal G_1 \to \mathcal G_2\to \mathcal G_3\to \mathcal G_1 \to \mathcal G_2\to\cdots,\end{equation}
and at each time step, there is a switching in the neighbor graphs; that is, the dwelling time for each active graph is $T_d=h=0.2$.
For the initial state, we choose the initial position $x(0)\in (\mathbb R^3)^3$ with nine coordinates randomly distributed in an  interval of length 10; while the coordinates of initial velocities were  randomly chosen from an interval of unit length. In Figure \ref{figsimul} we show the evolution of relative positions  $\hat x_i=x_i-x_3$ and the evolution of the norm of total relative velocity  $(|\hat v_1|^2+|\hat v_2|^2)^{\frac{1}{2}}$. These simulations show the asymptotic (exponentially fast) flocking behavior of the C-S model with alternating leaders.

In Figure \ref{figsimul2} (a), we  also exhibit the evolution of  velocities $v_1$, $v_2$, and  $v_3$ under the switching signal $\sigma$.  Here, we use different colors to denote different  coordinates, and use different markers to describe different agents.
%; for example, the blue lines express the evolution of their first  coordinates
%($x$-coordinates) and the star marker lines express the evolution of the velocity of the  agent 1.
 Figure \ref{figsimul2} (a) indicates that each coordinate of their velocities asymptotically attains an alignment, i.e., they exhibit a velocity consensus asymptotically.   To compare the relaxation process of the flocking state under different switching topology, we did some simulations with different switching signal.  We chose the similar signal with a sequence \eqref{switchsig}  but with a different dwelling time $T_d$ for each neighbor graph.  Precisely, we take $T_d=1$ seconds, $T_d=3 $ seconds and  $T_d=7 $ seconds for the simulations exhibited in Figure \ref{figsimul2} (b), (c) and (d), respectively.  We observe that in any case the velocities   attain a near alignment after  $t_0=9$ and keep like this. This means that the flocking state is robust about the alternating leaders.  Moreover, we observe that the switching signal before the alignment influence the asymptotic value a lot.  Since the leader in  the first active interaction topology is agent 1, if the dwelling time is longer, the asymptotic velocity is more close to the initial value of the agent 1.

\begin{figure}[h]
\centering
\subfigure[The evolution of each coordinate  of relative positions  $\hat x^k_i=x^k_i-x^k_3,\,\, i=1,2, \,k=1,2,3. $]
{\includegraphics[width=7cm]{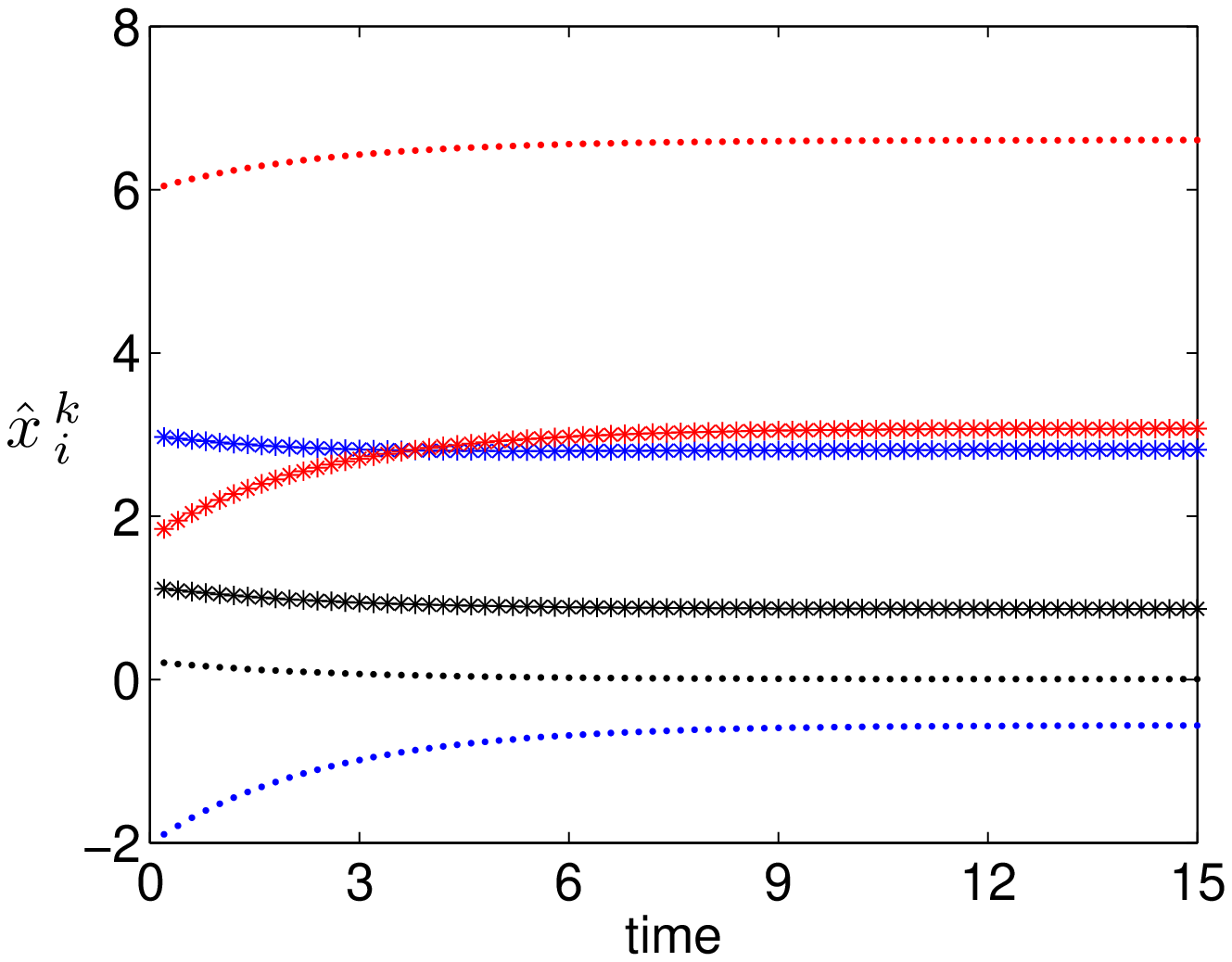}}
\subfigure[The evolution of  $(|\hat v_1|^2+|\hat v_2|^2)^\frac{1}{2}$ for the  relative velocities $\hat v_i=v_i-v_3, \,\, i=1,2.$]
{\includegraphics[width=8cm]{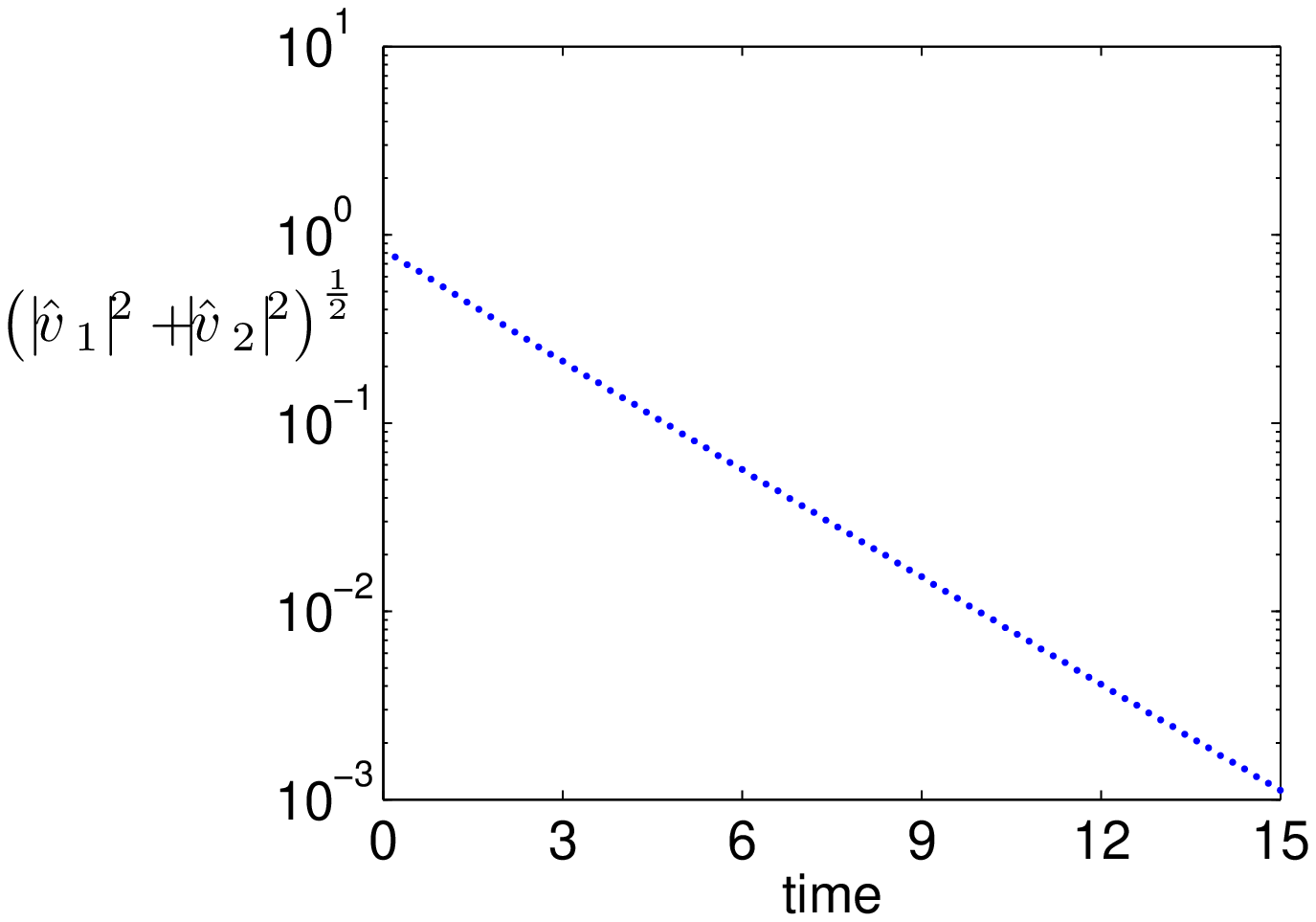}}
\caption{The emergence of flocking in a 3-agent C-S model with alternating leaders}\label{figsimul}
\end{figure}

\begin{figure}[h]
\centering
\subfigure[$T_d=h=0.2$]
{\includegraphics[width=7cm]{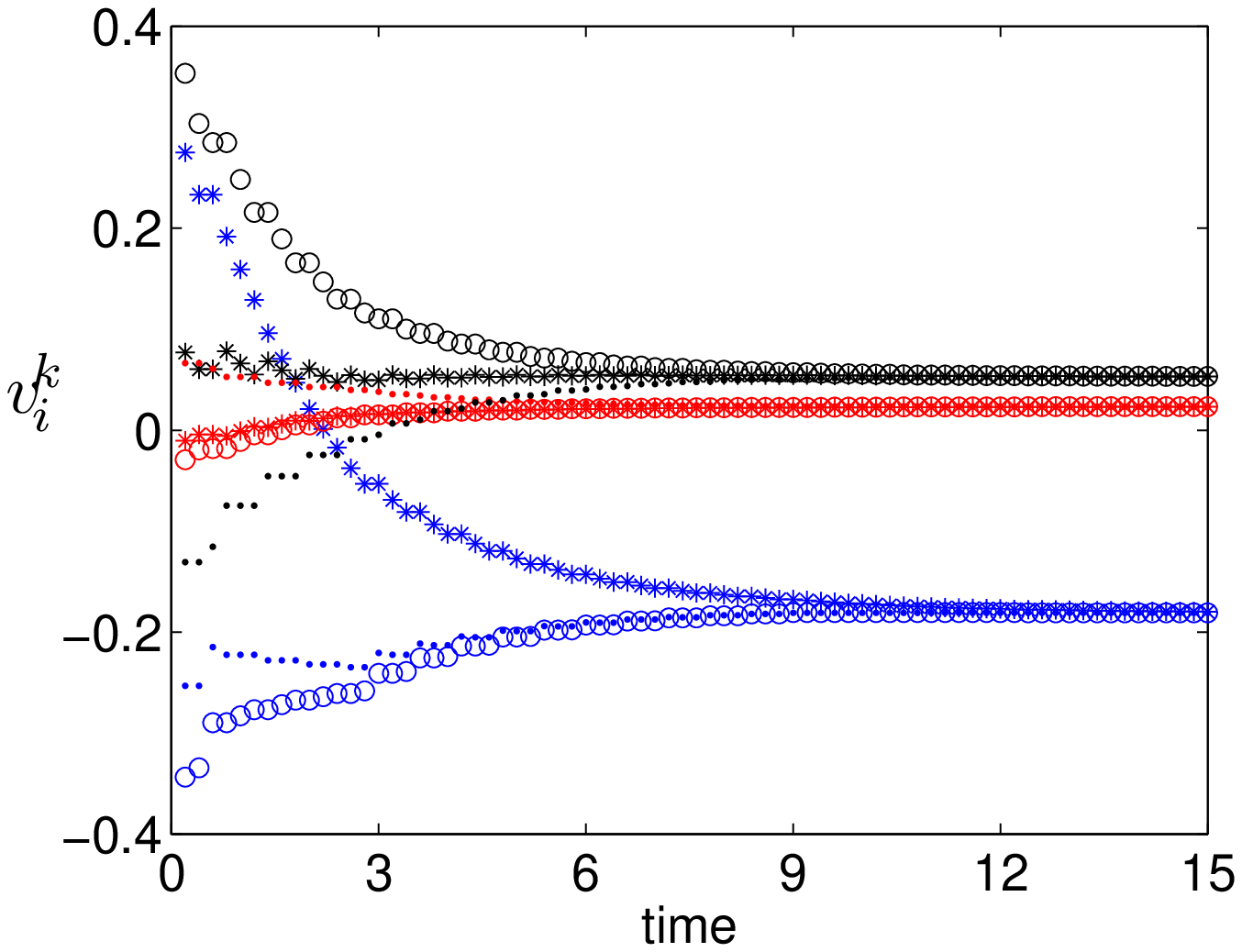}}
\subfigure[$T_d=1$]
{\includegraphics[width=7cm]{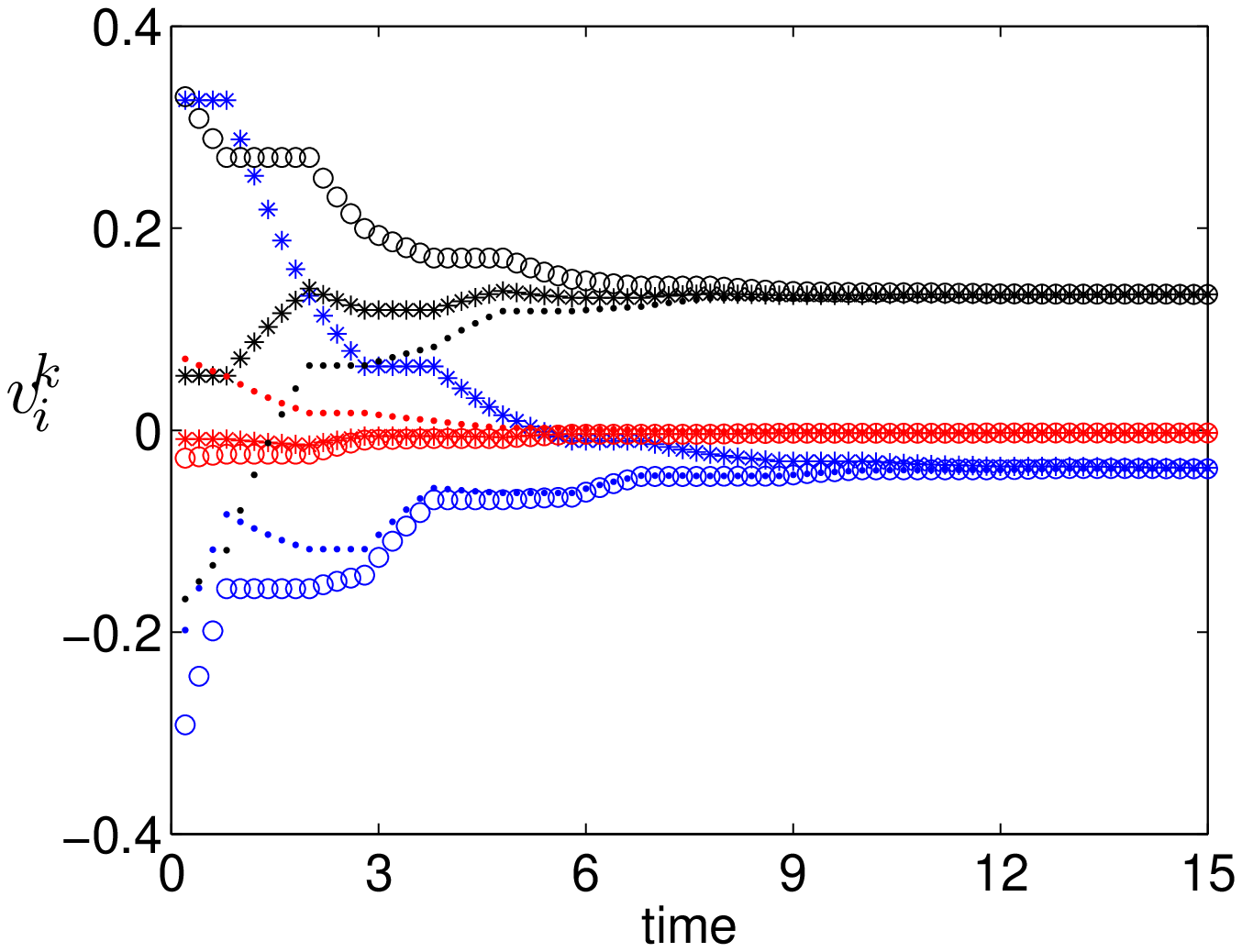}}
\subfigure[$T_d=3$]
{\includegraphics[width=7cm]{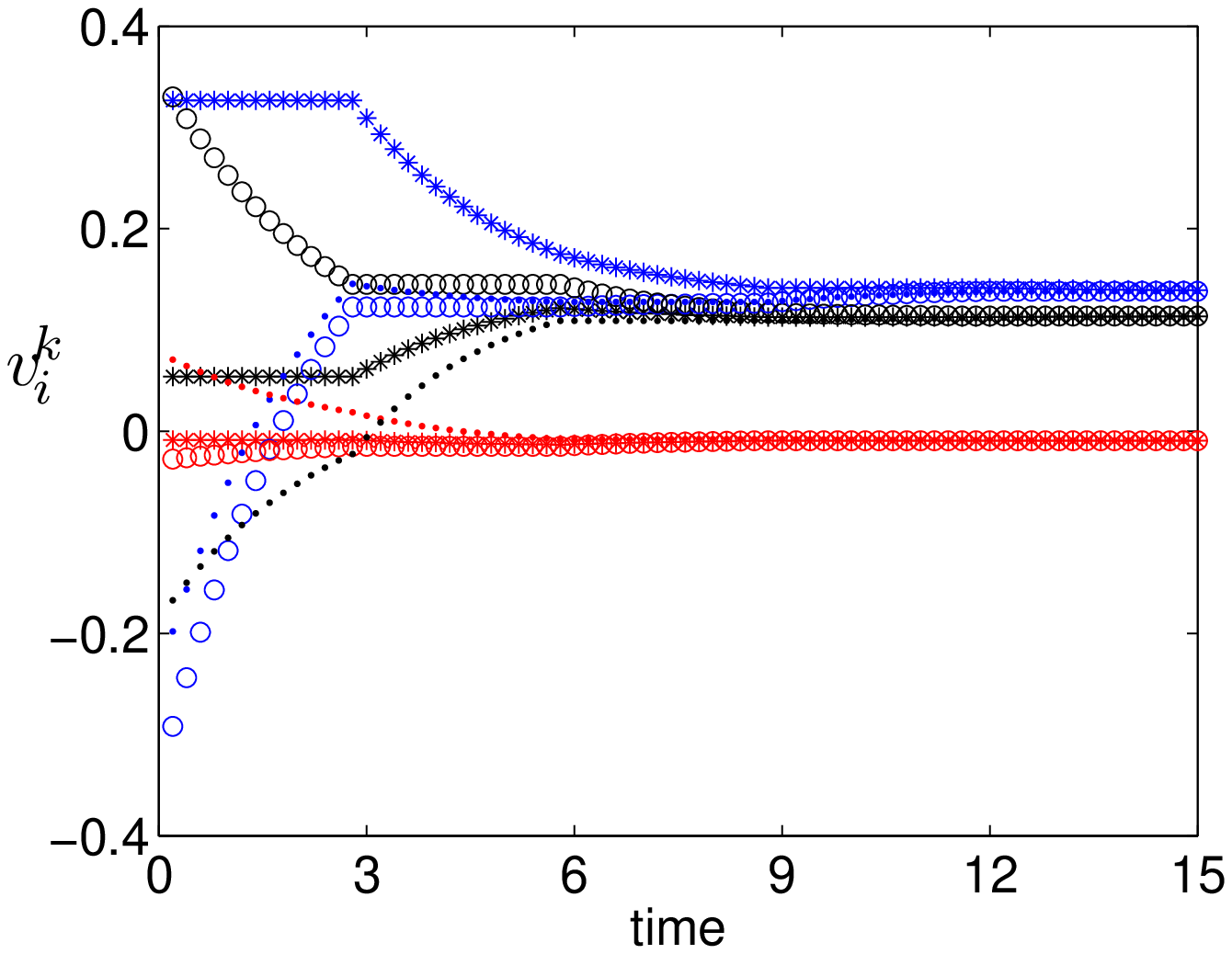}}
\subfigure[$T_d=7$]
{\includegraphics[width=7cm]{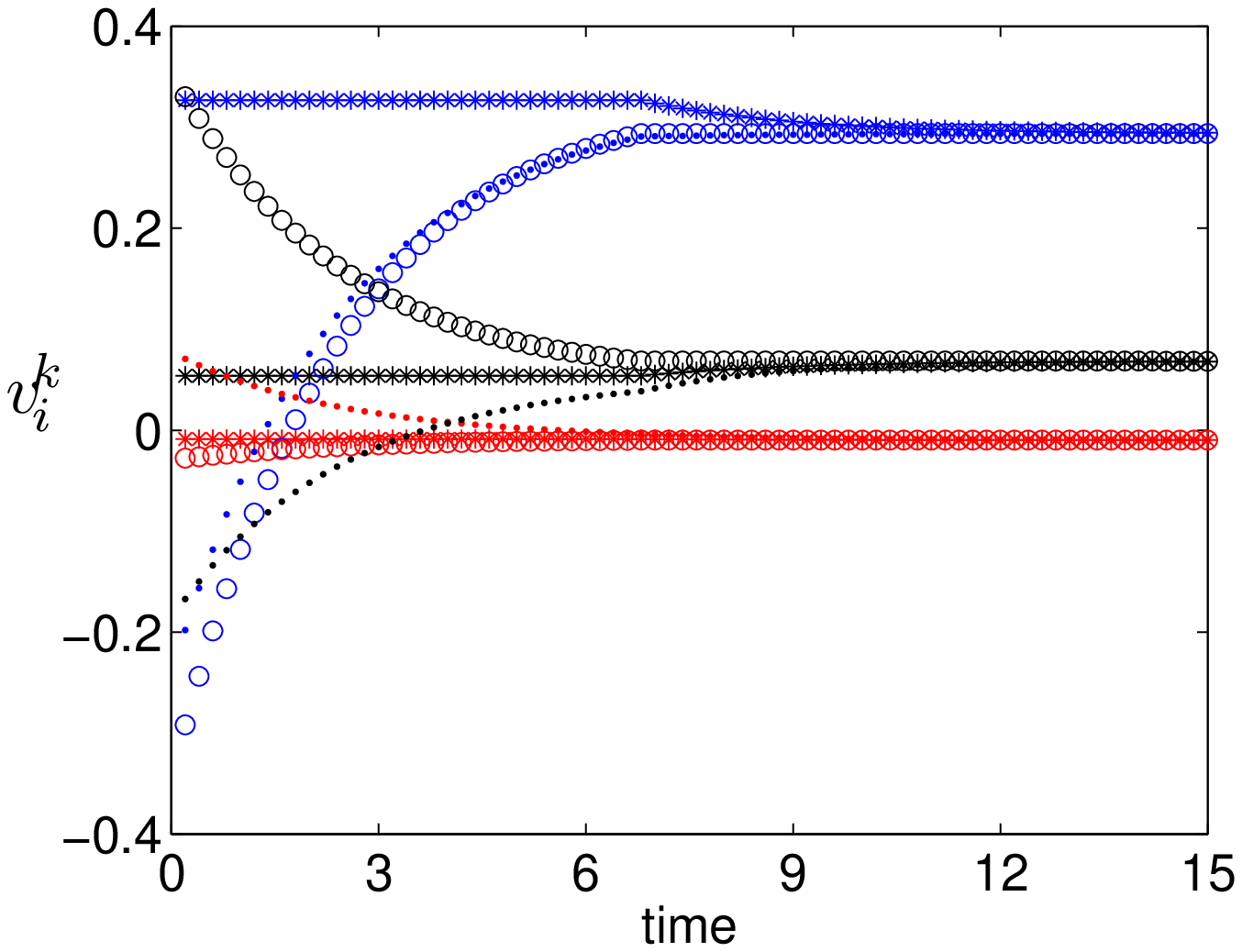}}
\caption{The  velocity alignment in a 3-agent C-S model with alternating leaders. The lines in blue, red and black indicate  $x$, $y$ and $z$ coordinates, respectively.  The lines marked by star, circle and dot indicate the agent 1, 2 and 3, respectively. }\label{figsimul2}
\end{figure}

\section{Conclusion}
We  studied the Cucker-Smale flocking under the rooted leadership with alternating leaders. This dynamically changing interaction topology is motivated by the ubiquitous phenomena in our nature such as the alternating leaders in migratory birds on the long journey,  the changing political leaders in human societies, etc.
Our study showed that the flocking behavior can occur for such a dynamically changed leadership structure under some sufficient conditions on the initial configurations depending on the decay rate of communications and the size of flocking.

%\section*{Acknowledgment}
%This work was carried out when the first author visited Seoul National University; he appreciated the support from PARC-SNU. The work of Z. Li was supported by  KRF-2009-0093137, NSFC grants 11271099, and
%the Fundamental Research Funds for the Central Universities grant HIT.NSRIF.2011002. The work of S.-Y. Ha was partially supported by KRF-2011-0015388.

\ifCLASSOPTIONcaptionsoff
  \newpage
\fi

\end{document}